\documentclass{article}
\usepackage[utf8]{inputenc}
\usepackage{fullpage}
\usepackage{array}
\usepackage{makecell}

\usepackage{xcolor,colortbl}

\usepackage[
  colorlinks,
  linkcolor = blue,
  citecolor = blue,
  urlcolor = blue]{hyperref}

\usepackage{graphicx}
\newcommand{\poly}{\textnormal{poly}}
\usepackage{amsthm, amsmath, mathtools, amsfonts, enumerate}
\usepackage[scr=rsfs,scrscaled=1]{mathalpha}
\usepackage{pythonhighlight}
\usepackage[T1]{fontenc}
\usepackage{tikz}
\usetikzlibrary{quantikz}

\theoremstyle{plain}
\newtheorem{theorem}{Theorem}[section]

\newtheorem{lemma}[theorem]{Lemma}
\newtheorem{corollary}[theorem]{Corollary}

\newtheorem{definition}[theorem]{Definition}
\newtheorem{prop}[theorem]{Proposition}

\newtheorem{remark}[theorem]{Remark}
\newtheorem{proposition}[theorem]{Proposition}

\newtheorem*{theorem*}{Theorem}
\newtheorem*{problem*}{Problem}

\theoremstyle{definition}

\usepackage{cleveref}
\crefname{theorem}{Theorem}{Theorems}
\crefname{lemma}{Lemma}{Lemmas}
\crefname{proposition}{Proposition}{Propositions}
\crefname{definition}{Definition}{Definitions}
\crefname{corollary}{Corollary}{Corollaries}
\crefname{example}{Example}{Examples}
\crefname{section}{Section}{Sections}
\crefname{appendix}{Appendix}{Appendices}
\crefname{table}{Table}{Tables}
\crefname{conjecture}{Conjecture}{Conjectures}

\usepackage{xcolor}

\usepackage{dsfont} \usepackage{mathrsfs} \usepackage{microtype}

\usepackage{thmtools,thm-restate}

\renewcommand{\P}[2][]{{\textnormal{Pr}_{#1}}{\left[#2\right]}}
\newcommand{\E}[2][]{{\textnormal{E}_{#1}}{\left[#2\right]}}

\newcommand{\indic}[1]{\delta\sof[\big]{#1}}

\newcommand{\R}{\mathbb{R}}
\newcommand{\C}{\mathbb{C}}
\newcommand{\Z}{\mathbb{Z}}

\newcommand{\eqdef}{:=}

\newcommand{\one}{\mathds{1}}

\newcommand{\Stab}{\textnormal{Stab}}
\newcommand{\Cliff}{\textnormal{Cliff}}

\usepackage{mathtools}
\mathtoolsset{centercolon}
\DeclarePairedDelimiter{\set}{\lbrace}{\rbrace}
\DeclarePairedDelimiter{\abs}{\lvert}{\rvert}
\DeclarePairedDelimiter{\card}{\lvert}{\rvert} \DeclarePairedDelimiter{\norm}{\lVert}{\rVert}

\DeclarePairedDelimiter{\pr}{\lparen}{\rparen} \DeclarePairedDelimiter{\sof}{\lbrack}{\rbrack}

\usepackage{authblk}

\newcommand{\F}{\mathbb{F}_2}

\title{Improved bounds for testing low stabilizer complexity states}
\author{Saeed Mehraban \footnote{Tufts CS,  Medford MA, Saeed.Mehraban@tufts.edu}\quad Mehrdad Tahmasbi\footnote{UIUC CS, Champaign IL, mehrdad@illinois.edu}}
\date{\today}

\begin{document}

\maketitle
\begin{abstract}
    Stabilizer states are fundamental families of quantum states with crucial applications such as error correction, quantum computation, and simulation of quantum circuits. In this paper, we study the problem of testing how close or far a quantum state is to a stabilizer state. We make two contributions: First, we improve the state-of-the-art parameters for the tolerant testing of stabilizer states. In particular, we show that there is an efficient quantum primitive to distinguish if the maximum fidelity of a quantum state with a stabilizer state is $\geq \epsilon_1$ or $\leq \epsilon_2$, given one of them is the case, provided that $\epsilon_2 \leq \epsilon_1^{O(1)}$. This result improves the parameters in the previous work \cite{arunachalam2024tolerant} which assumed $\epsilon_2 \leq e^{- 1/\epsilon^{O(1)}_1}$ \cite{arunachalam2024tolerant}. Our proof technique extends the toolsets developed in \cite{arunachalam2024tolerant} by applying a random Clifford map which balances the characteristic function of a quantum state, enabling the use of standard proof techniques from higher-order Fourier analysis for Boolean functions \cite{hatami2019higher, Samorodnitsky2007}, where improved testing bounds are available.
    Second, we study the problem of testing low stabilizer rank states. We show that if for an infinite family of quantum states stabilizer rank is lower than a constant independent of system size, then stabilizer fidelity is lower bounded by an absolute constant. Using a result of \cite{grewal2022low}, one of the implications of this result is that low approximate stabilizer rank states are not pseudo-random. 

    At the same time our work was completed and posted on arXiv, two other groups \cite{bao2024tolerant, arunachalam2024note} independently achieved similar exponential to polynomial improvements for tolerant testing, each using a different approach.
    
\end{abstract}

\tableofcontents

\section{Introduction}

Standard results in quantum state tomography \cite{haah2016sample, o2016efficient} indicate that to fully determine the description of a quantum state, the number of copies needed scales with the dimension of the Hilbert space. This dimension grows exponentially with the number of parts (e.g., qubits) in the system, making full tomography inefficient for many applications. However, if we are interested in testing specific properties of quantum states,  far fewer copies may be sufficient. For instance, testing if the rank of a density matrix is a given constant requires only a constant number of copies \cite{o2015quantum}. Testing specific features of a mathematical problem is the subject of the field of ``property testing'' which has had a fundamental success in theoretical computer science by designing fast algorithms that test vast amounts of data in various settings \cite{blum1990self, goldreich1998property, fischer2004art, ron2008property}. More recently, increasing interest has been dedicated to designing property testing protocols for problems in quantum information science (See \cite{montanaro2013survey} for a review). 

In this paper, we study the problem of testing quantum states with low stabilizer complexity. Stabilizer states are fundamental families of quantum states that appear in many applications such as simulating quantum computers on classical computers \cite{gottesman1998heisenberg}, quantum error correction \cite{gottesman1997stabilizer}, or defining models of quantum computation \cite{briegel2009measurement}. 
Determining whether a quantum state is close to or far from a stabilizer state can help us understand whether the state inherits key features of stabilizer states. As a result, it is important to have an efficient procedure to carry out this task. 

There are multiple ways to measure how close or far a quantum state is from the stabilizer ensemble. For example, \emph{stabilizer fidelity} is defined as the maximum overlap between a stabilizer state and the given quantum state.
\emph{Approximate stabilizer rank} is the minimum number of stabilizer states in any approximate decomposition of the given state into stabilizer states.
These different definitions may not always be comparable, as each can be more suitable for specific applications. For instance, stabilizer rank is closely tied to the computational complexity of simulating quantum systems, while stabilizer fidelity is more relevant to applications like distinguishing between quantum states.

An important question is: How many samples are needed to test the closeness or distance of a quantum state to the set of stabilizer states based on one of these metrics? The formulation and solution of this problem crucially depend on the metric we choose. 
Recent results have made important progress in addressing this problem. In \cite{montanaro2017learning} it was shown that $O(n)$ copies of an unknown stabilizer state over $n$ qubits are sufficient to identify the specific stabilizer state. Crucially this result is based on a measurement primitive called bell sampling which is the basis for several subsequent works. \cite{gross2021schur} showed that ``six'' copies of a quantum state are sufficient to test if that quantum state is exactly a stabilizer state or it is far from the stabilizer family in terms of fidelity. Their measurement protocol is a generalization of Bell sampling, which they call Bell ``difference'' sampling. They interpreted this primitive through a variant of Schur–-Weyl duality for the Clifford group, which they developed as their main technical development. \cite{grewal2022low} showed that using $O(k^{12})$ copies of a quantum state we can distinguish either if the quantum state has a fidelity of $\geq 1/k$ with a stabilizer state or if it is a Haar random state (which has an exponentially small overlap with any of the stabilizer states), thus demonstrating that high stabilizer fidelity states are not pseudorandom. In a subsequent work the same authors \cite{Grewal2024improved} showed that we can find a stabilizer with overlap $\epsilon - \tau$ for any state of stabilizer fidelity $\epsilon$ given polynomial copies of the state but exponential post processing time. This result were later improved in \cite{chen2024stabilizerbootstrappingrecipeefficient} to have polynomial post-processing time.   \cite{Grewal2024improved} also considered ``tolerant testing of stabilizer states'', i.e., the problem of deciding if stabilizer fidelity of a quantum state is $\geq \epsilon_1$ or $\leq \epsilon_2$. They showed that this task is possible given $\epsilon_2 \leq \frac{4\epsilon_1^6 - 1}{3}$ with $O(\poly(1/\epsilon_1))$ many samples using Bell difference sampling; in particular this result has the limitation that it assumes $\epsilon_1 \geq \sqrt[3]{1/2}$.   \cite{arunachalam2024tolerant} made further progress on tolerant stabilizer testing by showed that assuming a conjecture in additive combinatorics, we can do tolerant testing given $\epsilon_2 \leq 2^{-\text{poly}(1/\epsilon_1)}$ with $O(\poly(1/\epsilon_1))$ many samples. For phase states, they prove this result assuming $\epsilon_2 = \epsilon_1^{O(1)}$. The complexity of their procedure is poly$(n/\epsilon_1)$, where $n$ is the number of qubits. 

A major technical tool used in \cite{arunachalam2024tolerant} is the Gowers norm, which has been applied in additive combinatorics and higher-order Fourier analysis \cite{hatami2019higher, Gowers_2001, tao2012higher} to characterize phase structure in complex-valued functions defined on groups.  In particular for a function $f : \F^n \rightarrow \C$ the Gowers 3-norm defined as 
\begin{multline}
        \norm{f}_{U^3}^8 = \frac{1}{16^n} \sum_{x, h_1, h_2, h_3\in\F^n} f(x) \overline{f(x+h_1) f(x+h_2) f(x+h_3)} \\
        \times f(x+h_1+h_2)f(x+h_1 + h_3) f(x+h_2+h_3) \overline{f(x+h_1+h_2+h_3)}.
\end{multline}
More detailed definition is provided in \cref{sec:measures}. We can show that for $f : \F^n \rightarrow \{+1,-1\}$ if $\|f\|_{U^3}$ is large, then $f$ has a high correlation with a quadratic phase and if it is small it is far (See Theorem 5.3. in \cite{hatami2019higher}). This definition is particularly useful because stabilizer states have quadratic phase structure \cite{VanDenNest_2010}. 

In this work, we present the following advancements in this line of research: (1) We improve \cite{arunachalam2024tolerant} by showing that tolerant testing of quantum states is possible assuming $\epsilon_2 = \epsilon_1^{O(1)}$ in time $\text{poly} (n/\epsilon_1)$ without relying on additional conjectures. (2) We show that for an infinite family of quantum states with an approximate stabilizer rank bounded from above by a constant, the stabilizer rank of any state from that family is bounded below by a constant. As an implication of this result, we show that (approximate) low stabilizer rank states are not pseudorandom, generalizing \cite{grewal2022low}.

The organization of this paper is as follows. In \cref{sec:intro-results} we present the main results of this paper. In \cref{sec:intro-proof-ideas} and \cref{sec:future} we describe proof ideas and future directions. In \cref{sec:preliminaries} we discuss preliminaries and notations. In \cref{sec:proof1} we discuss the proof of the improved bounds for tolerant testing. \cref{sec:proof2} is dedicated to the proof of the second main result about the relationship between stabilizer fidelity and stabilizer rank. The remainder of this paper consists of appendices.


\subsection{Main results}
\label{sec:intro-results}
Our first result is improving the parameters for tolerant testing of stabilizer states in \cite{arunachalam2024tolerant}. For a quantum state $\ket{\phi} = \frac{1}{\sqrt{N}} \sum_x g(x) \ket{x}$ we define the Gowers 3-norm of the quantum state to be $\norm{\ket{\phi}}_{U^3} := \norm{g}_{U^3}$. 
\begin{theorem}
\label{th:main1} [Relating Gowers 3-norm to fidelity]
    Let $\ket{\phi}$ be a quantum state with $\norm{\ket{\phi}}_{U^3}^{8} \geq \gamma $. Then, there exists a stabilizer state $\ket{s}$ such that $\abs{\braket{\phi}{s}} \geq \frac{\gamma^{C_2}}{C_1}$ where $C_1, C_2 > 0$ are two absolute constants 
\end{theorem}
\noindent (See Remark~\ref{rem:constants} for the values of $C_1$ and $C_2$). 
Using Bell difference sampling introduced in \cite{gross2021schur} (also used in \cite{arunachalam2024tolerant}; see Lemma 4.8 therein) we can show that for a quantum state $\ket{\phi} \in (\mathbb{C}^2)^{\otimes n}$ we can estimate a quantity $R$ such that $\norm{\ket{\phi}}^{16}_{U^3} \leq R \leq \norm{\ket{\phi}}^8_{U^3}$ to within $\delta$ additive error using $O(1/\delta^2)$ copies of $\ket{\phi}$ and a (uniform) circuit of size $O(n/\delta^2)$. See also \cref{sec:bell-sampling} for an overview.
Combining this result with Theorem~\ref{th:main1}, we can show the following result. Let $Stab_n$ be the set of stabilizer states over $n$ qubits, and stabilizer fidelity be $F(\phi) = \max_{\ket{s} \in Stab_n} |\braket{s}{\phi}|^2$. 
\begin{corollary} [Improved bounds on tolerant testing]
Let $\ket{\phi} \in (\mathbb{C}^2)^{\otimes n}$ with the promise that either $F(\phi) \geq \epsilon_1$ or $\leq \epsilon_2$. Using $\poly(\epsilon_1)$ copies of a quantum state $\ket{\phi} \in (\mathbb{C}^2)^{\otimes n}$ and a circuit of size $n \cdot \poly(1/\epsilon_1)$ we can distinguish between the two cases with probability of error $\leq 1/3$ provided that $\epsilon_2 \leq \epsilon_1^C$ for a sufficiently large absolute constant $C > 0$.     
\end{corollary}

 Our second main result proves a relationship between stabilizer fidelity and approximate stabilizer rank with applications to distinguishing low stabilizer rank states from Haar states. For a quantum state $\ket{\phi}$, $\delta$-approximate stabilizer rank $\chi_\delta (\phi)$ is the minimum number $r$ such that there exist stabilizer states $\ket{s_1}, \ldots, \ket{s_r}$ and $c_1, \ldots, c_r \in \mathbb{C}$ such that $\|\ket{\phi} - c_1 \ket{s_1} + \ldots + c_r \ket{s_r}\| \leq \delta$. For exact rank (i.e., $\chi_0$), we use the notation $\chi$.
\begin{theorem} [Relating stabilizer rank and fidelity]
\label{th:main2}
    For each $k>0$ there exists $\delta_k > 0$ such that the following hold. Let $\ket{\phi}$ be a quantum state with $\chi(\phi) \leq k$. Then, $F(\phi) \geq \delta_k$, where $\delta_k$ is a number that only depends on $k$.
\end{theorem}
\noindent We note that our result only proves the existence of $\delta_k$; we leave a more concrete relationship to future work. We conjecture that $\delta_k \geq 2^{- k}/C$ for an absolute constant $C$. The above result holds for approximate rank $\chi_\delta$ for any $0 \leq \delta < \delta_k$.  
Previously, it was known \cite{mehraban2023quadratic, Labib2022stabilizerrank} that for a quantum state $\ket{\phi} = \frac{1}{\sqrt{N}} \sum_x g(x) \ket{x}$, $\chi(\phi) \geq \frac{2}{3} \log (\frac{\alpha^2}{\beta \sqrt{F(\phi)}})$, where $\alpha = \min_x |g(x)|$, $\beta = \max_x |g(x)|$. Therefore, the bound is not useful when $\ket \phi$ does not have full support or has large amplitudes. We note that the result does not hold in the other direction since we can always find quantum states with high stabilizer fidelity and high stabilizer rank. For instance, let $\epsilon \ll 1$, $\ket{s} \in \Stab_n$ and $\ket{\psi}$ be a quantum state with high stabilizer rank, e.g., a Haar state which has approximate rank $\geq 2^{n}/\poly(n)$ \cite{mehraban2023quadratic}. Then for any $k$ and $\delta_1, \delta_2$ we can choose $\epsilon$ such that (after normalization) $\sqrt{1-\epsilon} \ket{s} + \sqrt{\epsilon} \ket{\phi}$ has fidelity $\geq 1- \delta_1$ and $\delta_2$-approximate rank $\geq k$. For more details, see \cref{sec:other-rel}. 

Using Bell difference sampling (performed in \cite{grewal2022low}; see also Appendix~\ref{sec:bell-sampling}) we obtain the following implication.

\begin{corollary} [Low stabilizer rank states are not pseudorandom]
\label{cor:main2}
    For each constant $k$, there exists a quantum algorithm with gate complexity $\poly(n)$ that given $O(1)$ copies of a quantum $\ket{\phi}$ can distinguish cases
    \begin{enumerate}
        \item $\chi(\phi) \leq k$
        \item $\ket{\phi}$ is sampled from Haar measure
    \end{enumerate}
    with probability of error less than $1/3$.
\end{corollary}

\subsection{Proof ideas}
\label{sec:intro-proof-ideas}

\subsubsection{Ideas for Theorem \ref{th:main1}}
We first outline the proof ideas for the case of a phase state, i.e., $\ket{\phi} = \frac{1}{\sqrt{N}}\sum_{x} (-1)^{f(x)}\ket{x}$, which was proven before \cite{arunachalam2024tolerant} (weaker quasipolynomial bound was also suggested in remark A.2. of \cite{mehraban2023quadratic}). We then explain what limits us in proving the main result in Theorem~\ref{th:main1} and how we improve on the parameters of \cite{arunachalam2024tolerant} for non-phase states.

\paragraph{Ideas for phase states:}
For the special case of the phase state, the proof follows closely that of the inverse Gowers theorem \cite{hatami2019higher, Samorodnitsky2007} for classical functions and has two high-level ideas (based on \cite{arunachalam2024tolerant}):
\paragraph{(1) Probabilistic construction of an almost linear function}
Consider the usual Pauli generators $X^a$ and $Z^b$, for $a, b \in \F^n$ which satisfy $X^a \ket{x} = \ket{x + a}$ (addition mod $2$) and $Z^b \ket {x} = (-1)^{\langle{x,\alpha}\rangle} \ket{x}$, where $\langle x, \alpha \rangle = \sum_i x_i \alpha_i \mod 2$ is the inner product between $x$ and $\alpha$.
Using a probabilistic construction, we construct a function $\zeta : \F^n \to \F^n$ such that 1) ${\abs{\bra{\phi} X^y Z^{\zeta(y)} \ket{\phi}}^2} = \poly(\gamma)$ for $\poly(\gamma)$ fraction of $y \in \F^n$  and 2) $\zeta(x+y) = \zeta(x) + \zeta(y)$ with probability $\poly(\gamma)$ for randomly chosen $x$ and $y$. The probabilistic construction is as follows: We choose $\zeta(y) = \alpha$ with probability $\abs{\bra{\phi} X^y Z^{\alpha} \ket{\phi}}^2$ and independently from other $y$s. \emph{The main bottleneck in going beyond phase states is that this is a valid probability distribution only for phase states.}

\paragraph{(2) Ideas from additive combinatorics} Two techniques from additive combinatorics are used to extend the almost-linear function $\zeta$ to a linear function. See \cref{sec:background-additive-comb} for more details about these two tools. First, by applying the Balog-Szemerédi-Gowers theorem \cite{Balog1994AST}, we identify a subset $S$ of the graph of $\zeta$ with bounded doubling size (i.e., bounding $|S + S| \leq K |S|$, where $K$ is not too large). Next, a recent result by Gowers, Green, Manners, and Tao on Marton's conjecture \cite{gowers2023conjecturemarton} implies that any set with bounded doubling size (relative to its size) can be covered by a few translates of a linear subspace. Moreover, we can bound the size of this linear subspace. Finally, standard linear algebra allows us to cover this subspace using graphs of affine linear maps.

Combining these ideas, one get a linear map $\ell$ such that $\E[y]{\abs{\bra{\phi} X^y Z^{\ell(y)}\ket{\phi}}^2} = \poly(\gamma)$. There are standard techniques to find a quadratic phase state (which is the stabilizer state) whose overlap with $\ket{\phi}$ is $\poly(\gamma)$ from this result.

\paragraph{Ideas for non-phase states:}
When $\ket{\phi}$ is not a phase state, as discussed in \cite{arunachalam2024tolerant}, the proof for phase states faces several obstacles. In particular, contrary to phase states, which have overlap $\poly(\gamma)$ with a quadratic phase state,  computational basis states (which are stabilizers themselves) have overlap $O(1/\sqrt{N})$ with any quadratic phase state. On a technical level, if we set $f(y, \alpha) \eqdef \abs{\bra{\phi} X^{y} Z^{\alpha} \ket{\phi}}^2$, for phase state, we have $\sum_\alpha f(y, \alpha) = 1$ for all $y$, which is crucial in the proof. Nevertheless, for arbitrary state $\ket{\phi}$, the function $\sum_\alpha f(y, \alpha)$ can be concentrated on few $y$s. 

In \cite{arunachalam2024tolerant}, the authors obtain a weaker result for general states using different proof techniques. Instead of choosing a subset of all $(y, \alpha)$ that corresponds to the graph of almost linear maps, they sample each $(y, \alpha)$ with probability $f(y, \alpha)$ and independence from other $(y, \alpha)$s. They can show that this set is almost linear and can be ``linearized'' using ideas from additive combinatorics; in doing so, they have to consider a conjecture. However, the linear subspace they obtain might have $\poly(\gamma^{-1})$ non-commuting elements. Therefore, they need $\exp(\poly(\gamma^{-1}))$ subspaces that all elements are commuting to cover this linear subspace. Then, they show that the stabilizer corresponding to one of these subspaces has overlap $\exp(-\poly(\gamma^{-1}))$ with the original state.

We circumvent the difficulty in generalizing the phase state proof with another idea. Our proof gives a better bound and does not rely on any conjecture. We apply a random Clifford to our state, which has two properties: 1) all stabilizer measures are invariant under Clifford operations, and 2) They form a 3-design. The first property implies that showing Theorem~\ref{th:main1} for a state $\ket{\phi}$ is equivalent to showing the Theorem~\ref{th:main1} for $C\ket{\phi}$ for any Clifford $C$. We also use the second property of Clifford operation to show that for some $C$, the function $f(y, \alpha)$ defined for $C\ket{\phi}$ is ``balanced'' over different values of $y$. More precisely, we show that $\sum_\alpha f(y, \alpha) = O(1)$ for all $y$. While our state is still not a phase state (in fact, we can have amplitudes $\omega(1)$ for a random Clifford), we show that we can modify the proof for the phase state to work in this situation. 

Independently from our work, two other papers \cite{arunachalam2024note} and \cite{bao2024tolerant} obtained the same exponential improvement of the  bounds in \cite{arunachalam2024tolerant}. They both use the main technique in \cite{arunachalam2024tolerant} to deal with non-phase states, but complement it with ideas from graph theory or algebraic properties of Pauli group to bound the number of non commuting elements in a subspace of $\F^{2n}$.

\subsubsection{Proof ideas for Theorem~\ref{th:main2}}
To prove Theorem~\ref{th:main2}, we consider an arbitrary state $\ket{\phi}$ in the span of $k$ stabilizer state. We first show that to prove the stabilizer fidelity of $\ket{\phi}$ is $\Omega_k(1)$, it is enough to demonstrate that the minimum eigenvalue of the Gram matrix of any $k$ stabilizer state is either zero or lower bounded by $\Omega_k(1)$. To show this, we consider a sequence of Gram matrices of $k$ stabilizer states (possibly have different numbers of qubits). We then exploit a compactness argument to find a convergent subsequence of this sequence of Gram matrices. Next, we use several structures of stabilizer states to show that the sequence element should be fixed after some point. As a result, their minimum eigenvalue cannot converge to zero.

\subsection{Future directions and open questions}
\label{sec:future}

One important direction for future research is to refine the constant factors in our results. Currently, many constants are either unspecified or far from optimal. For example, although achieving an exponential-to-polynomial improvement in the relationship between parameters in Theorem~\ref{th:main1} marks significant progress, the constant factor \( C_2 \) (the exponent of \( \gamma \)) in our proof is currently $266$. We believe there is substantial potential for improving these constants. Another instance is $\delta_k$ in Theorem~\ref{th:main2}, which is not explicitly. Our proof technique only shows the existence of $\delta_k$ independent of $n$ but does not prove an explicit relationship. We conjecture that $\delta_k > 2^{-k}/C$ for an absolute constant $C$.

Another major open problem is extending the ``tolerant testing'' framework to stabilizer rank. In particular, consider the problem:
\begin{problem*} [Tolerant testing of approximate stabilizer rank]
Let $\ket{\phi} \in (\mathbb{C}^2)^{\otimes n}$, $\delta > 0$ and $k_1 \ll k_2$ such that:
either (i) there exists a quantum state with approximate stabilizer rank $\leq k_1$ that is within Euclidean distance $\delta$ of $\ket \phi$ or (ii) there exists a quantum state with approximate stabilizer rank $\geq k_2$ within $\delta$ distance of this quantum state. How many copies are necessary and sufficient to distinguish between (i) and (ii)?   
\end{problem*}

Currently, the existing tools are only suitable for $k_1 = 1$. For instance, this problem is an instance of \cite{gross2021schur} when $\delta = 0$.  However, the problem already becomes challenging for $k_1 = 2$ and $k_2 = 2^n/n^2$ (even when $\delta = 0$). The main difficulty is that stabilizer states constitute an overcomplete basis, and it is not clear how to decode the information about the minimal number of stabilizer states in any decomposition based on the characteristic function of the quantum states. For instance, for a stabilizer state, we know that the nonzero elements of the corresponding characteristic polynomial correspond to a linear subspace with $2^n$ elements (i.e., Lagrangian). Preserving such structure is crucial in the tools we use in this paper (or those in \cite{arunachalam2024tolerant}). However, if a quantum state is a linear combination of two stabilizer states with nontrivial overlap, then the structure breaks. We emphasize that due to sharp discontinuities in exact rank, tolerant testing would be meaningful only if we consider approximate rank.

Due to the inverse theorem for Gowers 3-norm \cite{Samorodnitsky2007}, we get a lot of leverage for distinguishing high stabilizer fidelity states from the Haar ensemble. However, counterexamples are known for inverse theorems for higher Gowers norms (See Theorem 5.6. of \cite{hatami2019higher}), suggesting that higher phase quantum states would be plausible candidates for pseudorandom quantum states. However, we note that inverse theorems also exist for Gowers $k$-norms, $k \geq 4$ for the so-called non-classical polynomials \cite{green2011inverse}, which may have high overlap with phase states of the same degree. Can we show rigorously that explicit families of quantum states with higher degree (e.g., degree $d$, for $d \gg 1$) phase structures are pseudorandom? 

\subsection*{Acknowledgment}
S. M. and M. T. acknowledge funding provided by NSF CCF-2013062. S. M. and M. T. are grateful to Daniel Liang, Vishnu Iyer, Kasso Okoudjou, and Makrand Sinha for their insightful conversations.


\section{Preliminaries}
\label{sec:preliminaries}
\subsection{Notations}

For a positive integer $k$ we use the notation $[k] := \{1,2, \ldots, k\}$ to denote the set of integers from $1$ to $k$. For any non empty set $S$ and function $f:S\to \C$, we define $\E[x\in S]{f(x)} = \frac{1}{\card{S}} \sum_{x\in S} f(x)$. $\indic{\cdot}$ is the indicator function. For any function $f:X\to Y$, we define its graph as $G(f) \eqdef \set{(x, f(x)), x\in X}$. $\F$ is the finite field of order (size) two.  For a function $g : \F^n \rightarrow \C$ and $y \in \F^n$, we define $\Delta_y g(x) = g(x) \overline{g (x + y)}$ to be the phase derivative of $g$ in the $y$ direction.  In particular, for $y_1, \ldots, y_k \in \F^n$ we have $\Delta_{y_k} \ldots \Delta_{y_1} g(x) = \prod_{S \subseteq [k]} g_S(x + \sum_{j \in S} y_j)$, where for every $x$ $g_S (x) = \overline{ g (x)}$ if $|S|$ is odd and $g_S = g$ otherwise. 

We work with $n$ qubit states throughout the paper, which are unit vectors in the vector space spanned by elements of $\F^n$.  We use the convention $N\eqdef 2^n$.  An $n$ qubit orthogonal unitary is a real $N\times N$ matrix with $O^tO = \one$  where $O^t$ is the transpose of $O$.  The Haar measure $\mathcal{O}$ over orthogonal matrices is (the unique) probability distribution over orthogonal matrices invariant under matrix multiplication.

\subsection{Boolean linear algebra}
We work with vector spaces over finite field $\F$.  The standard $n$-dimensional vector space is $(\F^n, +)$.  We consider the standard basis, $e_1, \cdots, e_n$,  for $\F^n$ where $e_j = 0^{j-1} 1 0^{n-j}$.  We use the inner-product $\langle x, y\rangle = \sum_{i=1}^n x_iy_i \mod 2$ for $x, y\in\F^n$. A \emph{subspace} $V$ is a subset of a larger vector space closed under addition and scalar multiplication (which implies $0 \in V$).  An \emph{affine subspace} is a subset of a vector space of the form $z + V$ for a vector $z$ and a subspace $V$.  A \emph{linear map} is a function between vector spaces preserving scalar multiplication and addition.  The \emph{transpose} of a linear map $\ell:\F^n \to \F^n$ is the unique linear  map $\ell^t:\F^n \to \F^n$ satisfying $\langle x, \ell(y) \rangle = \langle \ell^t(x), y \rangle$ for all $x, y\in \F^n$. A linear map $\ell:\F^n \to \F^n$ is called \emph{symmetric} if $\ell= \ell^t$. The \emph{diagonal} of a linear map $\ell:\F^n \to \F^n$ is a vector $d\in \F^n$ with $d_i = \langle e_i, \ell(e_i) \rangle$. In other words, if we consider the matrix entries $\ell_{i,j}$ for $\ell$ such that $\ell (e_j) = \sum_i\ell_{ij} e_i$, then $d_i = \ell_{ii}$. An \emph{affine linear map} is a function between vector spaces that is the sum of a linear map and a constant function. 

The following technical result will be useful in the proof of Theorem~\ref{th:main1} and will be proven in Appendix~\ref{sec:proof-linear-map}.

\begin{lemma}
    \label{lem:linear-map}
    Let $S\subset \F^{2n}$ and $V\subset \F^{2n}$ be an affine subspace. Let $U=\set{y\in \F^n: \exists y': (y, y')\in V}$. There exists an affine map $\ell:\F^n \to \F^n$ with $ \abs{G(\ell) \cap S} \geq \frac{\abs{S \cap V} \abs{U}}{\abs{V}}$.
\end{lemma}


\subsection{Analysis of Boolean functions}
For a function $g : \F^n \rightarrow \C$ its Fourier transform is the function $\hat{g} : \F^n \rightarrow \C$ with $\widehat{g }(\alpha) = \E[x\in \F^n] {(-1)^{\langle \alpha ,  x\rangle} g(x)}$. By Parseval's equality $\E[x\in\F^n]{|g(x)|^2} = \sum_{\alpha\in\F^n} |\widehat{g}(\alpha)|^2$. For two functions $f, g:\F^n \to\C$, we define $(f*g)(x) = \E[y\in\F^n]{f(y)g(x+y)}$. We have $\widehat{f * g}(\alpha) = \widehat{f}(\alpha) \widehat{g}(\alpha) $. By \cite[Claim 4.8]{hatami2019higher}, we  have for two real function $f,g:\F^n\to \R$
\begin{align}
\label{eq:fourier4}
    \sum_\alpha \widehat{f g}(\alpha)^4 
        &= \E[y]{\pr{\E[x]{\Delta_y f(x) \Delta_y g (x)}}^2}.
\end{align}

\subsection{Characteristic functions}
In this subsection, we review the definition of the characteristic function of a quantum state and its properties.  Let $\ket{\phi} = \frac{1}{\sqrt{N}} \sum_{x\in\F^n} g(x) \ket{x}$ be a vector in $\C^N$. We define the characteristic function $f_\phi:\F^{2n} \to \R$ of $\ket{\phi}$ with:
\begin{align}
    f_{\phi}(y, \alpha) = \abs{\widehat{\Delta_y g}(\alpha)}^2,
\end{align}
 for $y, \alpha \in \F^{n}$. Using quantum information language, we have $f_\phi(y, \alpha) = \abs{\bra{\phi} X^y Z^\alpha \ket{\phi}}^2$. When $\phi$ is clear from the context, we denote $f_\phi$ by $f$.  By Parseval's equality, we have (See \cref{a:basic-calculations} for a proof):
\begin{align}
    \frac{1}{N} \sum_{z\in\F^{2n}} f(z) = (\E[x]{\abs{g(x)}^2})^2 = \braket{\phi}{\phi}^2
    \label{eq:parseval-f}
\end{align}
For any characteristic function $f$ it holds that (adapted from Lemma 2.9 of \cite{arunachalam2024tolerant}; also Fact 3.2 of \cite{grewal2022low}) 
\begin{align}
    \label{eq:f-additive}
    \frac{1}{N^2} \sum_{z_1, z_2} f(z_1)f(z_2) f(z_1 + z_2) = \frac{1}{N}\sum_{z\in \F^{2n}} f(z)^3
\end{align}
We consider symplectic inner product on $\F^{2n}$ as $[(y_1, \alpha_1), (y_2, \alpha_2))] = \langle y_1, \alpha_2\rangle + \langle y_2, \alpha_1\rangle$ where $y_1, y_2, \alpha_1, \alpha_2 \in \F^n$ and $\langle \cdot, \cdot \rangle$ denote the standard inner product over $\F^n$. For any  $S\subset\F^{2n}$, we define $S^\perp \eqdef \set{z: [z, z'] = 0, \forall z' \in S}$. 
\begin{lemma}
\label{lem:phase-sum}
If $S$ is a linear subspace of $\F^{2n}$, then we have
\begin{align}
    \sum_{z\in S} (-1)^{[z, z']} = \abs{S} \cdot \indic{z' \in S^\perp}
    \label{eq:phase-sum}
\end{align}    
\end{lemma}

\begin{proof}
    If $z' \in S^\perp$, then the statement is trivially true.  Now suppose $z' \notin S^\perp$. Let $h : S \rightarrow \F$ with $h(z) = [z,z']$. Since $h$ is a linear function $\dim \ker (h) + \dim \text{Im} (h) = \dim (h)$. Since $z' \notin S^\perp$ there exists $z \in S$ such that $[z,z'] = 1$. Furthermore, there exist $z \in S$ (e.g., $z = 0^{2n}$) such that $[z,z'] = 0$. As a result, $\dim \text{Im} (h) = 1$. Therefore, $\dim \ker (h) = \dim (S) - 1$. Hence $|h^{-1}(0)| = |h^{-1} (1)| = |S|/2$ which implies $\sum_{z \in S} (-1)^{h(z)} = 0$.
\end{proof}

\begin{lemma} [Theorem 3.2. of \cite{grewal2022low}]
    For any $z\in \F^{2n}$ and any characteristic function $f$, the following identity holds
\begin{align}
    f(z) = \frac{1}{N}\sum_{z'\in \F^{2n}} (-1)^{[z, z']} f(z')
    \label{eq:f-sum}
\end{align}
relating $f$ to its Fourier transform with respect to the symplectic inner product.
\label{lem:f-sum}
\end{lemma}
Using the above lemma, we can prove the following lemma.
\begin{lemma}
\label{lem:f-sum-ineq}
    For any subspace $V\subset \F^{2n}$ and any $z'\in\F^{2n}$, 
    \begin{align}
        \sum_{z\in V} f(z) \geq \sum_{z\in V} f(z + z').
    \end{align}
\end{lemma}
\begin{proof}
    Consider an arbitrary $z' \in \F^{2n}$. It is enough to show that $\sum_{z\in V} f(z + z')$ is maximized when $z' = 0$.  Using Lemma~\ref{lem:phase-sum} and Lemma~\ref{lem:f-sum}, we rewrite this expression as 
    \begin{align}
        \sum_{z\in V} f(z + z') 
        &= \frac{1}{N}\sum_{z\in V} \sum_{t\in F^{2n}} (-1)^{[t, z+z']}f(t) \\
        &= \frac{1}{N}\sum_{t\in F^{2n}} f(t) (-1)^{[t, z']} \sum_{z\in V} (-1)^{[z, t]}\\
        &= \frac{\abs{V}}{N}\sum_{t\in V^\perp} f(t) (-1)^{[t, z']}
    \end{align}
    Since $f(t)$ is always non-negative, all terms in the summation above are non-negative when $z' = 0$.  Hence, the above expression takes its maximum value when $z' = 0$.
\end{proof}

\subsection{Stabilizer formalism}
\label{sec:stab-form}
Let $\mathcal{P}_n$ denote the Paulli group acting on $n$ qubits with  $\mathcal{P}_n := \{ i^c X^a Z^b, (a,b) \in \F^{2n}, c \in \Z_4\}$.  An $n$-qubit state $\ket{s}$ is called \emph{stabilizer state}, if there a exists an Abelian subgroup $S$ of size $2^n$ of $\mathcal{P}_n$ with $p\ket{s} = \ket{s}$ for all $p\in S$.   Let $\Stab_n$ be the set of all $n$-qubit stabilizer states.  Elements of $\Stab_n$ can be written of the form \cite{VanDenNest_2010}
\begin{align}
    \frac{1}{\sqrt{|A|} }\sum_{x\in A} i^{\ell(x)}(-1)^{Q(x)} \ket{x},
\end{align}
where $A\subset \F^n$ is  an affine subspace of $\F^n$ ($A = \{L y + v: y \in \F^m\}$, where $L \in \F^{n \times m}$, $V \in \F^n$), $\ell:\F^n \to \F$ is a linear function, and $Q:\F^n\to \F$ is a quadratic function.  

An $n$ qubit unitary $C$ is called \emph{Clifford} if for all $p\in \mathcal{P}_n$, $CpC^\dagger \in \mathcal{P}_n$. Let $\Cliff_n$ be the set of all Clifford operations on $n$ qubits.  We also need a subgroup of the Clifford group called \emph{real} Clifford group defined as unitaries generated by $Z$, $CNOT$, $H$ gates.  By definition, the entries of real Clifford unitaries are real in computational bases.  The following proposition is a known fact about stabilizer states and Clifford operations.
\begin{proposition}
        If $C$ is a Clifford unitary and $\ket{s}$ is a stabilizer state, then $C\ket{s}$ is a stabilizer state.
\end{proposition}
Additionally, we use the following result about $\Cliff^R_n$ \cite[Theorem~4]{Hashagen2018realrandomized}.
\begin{lemma}
    \label{lem:cliff-design}
    $\Cliff^R_n$ is an orthogonal 2-design, i.e., for any operator $\rho$ acting on $2n$ qubits, we have
        \begin{align}
            \E[C\in \Cliff^R]{ (C\otimes C) \rho (C^\dagger \otimes C^\dagger)} = \E[O \in \mathcal{O}]  {(O\otimes O) \rho (O^\dagger \otimes O^\dagger)}
        \end{align}
        where $\mathcal O$ is the orthogonal Haar distribution over $n$ qubits.
\end{lemma}

\subsection{Measures of stabilizer complexity}
\label{sec:measures}

Here, we overview three measures of stabilizer complexity.  Each measure characterizes how much an arbitrary state resembles a stabilizer state.
\paragraph{Stabilizer fidelity} For an $n$ qubit quantum state $\ket{\phi}$, we define its stabilizer fidelity as the maximum overlap with a stabilizer state, i.e., 
\begin{align}
    F(\ket{\phi}) \eqdef \max_{\ket{s}\in \Stab_n}\abs{\braket{s}{\phi}}^2.
\end{align}
We know that $F(\ket{\phi}) = 1$ if and only if $\ket{\phi}$ is a stabilizer state.
\paragraph{Stabilizer rank} The stabilizer rank of a quantum state is defined as the minimum number of stabilizer states needed in any linear decomposition of the state into stabilizer states.  More formally, we define
\begin{align}
    \chi(\ket{\phi}) = \min\set{r: \exists \ket{s_1}, \cdots, \ket{s_r}\in \Stab_n \text{ such that } \ket{\phi}\in \text{span}(\ket{s_1}, \cdots, \ket{s_r})}
\end{align}
This quantity plays an important role in the classical simulation of important families of quantum circuits \cite{Bravyi_Gosset_2016}.  We have $\chi(\ket{\phi}) = 1$ if and only if $\ket{\phi}$ is a stabilizer state. We also have $\chi(\ket{\phi}) \leq N$ for all quantum states $\ket{\phi}$. 
\paragraph{Gowers norm} Gowers norm has been introduced in the context of additive combinatorics and found applications in theoretical computer science, e.g., for property testing of classical functions \cite{hatami2019higher}. 
For a function $g:\F^n\to \C$, we define the Gowers $d$-norm as
\begin{align}
    \norm{g}_{U^d} \eqdef \pr{\E[x, y_1, \dots, y_d]{\Delta_{y_1}\cdots \Delta_{y_d} g(x)}}^{\frac{1}{2^d}}
\end{align}
Following \cite{mehraban2023quadratic, arunachalam2024tolerant}, if $\ket{\phi} = \frac{1}{\sqrt{N}}\sum_{x\in \F^n }g(x)\ket{x} $ is a vector in $\C^{N}$, we define $\norm{\ket{\phi}}_{U^d} \eqdef \norm{g}_{U^d}$ 

When $d=3$ Lemma 3.2. of \cite{arunachalam2024tolerant} implies that
\begin{align}
    \label{eq:gower-char}
    \norm{\ket{\phi}}_{U^3}^8 =  \E[y]{ \sum_\alpha \abs{\widehat{\Delta_y g}(\alpha)}^4} = \frac{1}{N} \sum_{z\in \F^{2n}}f^2(z)
\end{align}
where $f$ is the characteristic function of $\ket{\phi}$. Furthermore, by Theorem~3.4 of \cite{arunachalam2024tolerant}, we have $\norm{\ket{\phi}}_{U^3} =1 $ if and only if $\ket{\phi}$ is a stabilizer state.

\subsection{Additive combinatorics}
\label{sec:background-additive-comb}

We need the following two results from additive combinatorics proven in \cite{Balog1994AST} and \cite{gowers2023conjecturemarton}, respectively.
\begin{theorem}
\label{th:bgs}
    Let $(G,+)$ be an Abelian group and $S\subset G$ with $\P[z_1, z_2\in S]{z_1 + z_2 \in S} \geq \epsilon$. There exists, $S'\subset S$ with $\abs{S'} \geq \frac{\epsilon \abs{S}}{3}$ and $\abs{S' + S'} \leq \pr{\frac{6}{\epsilon}}^8\abs{S}$.
\end{theorem}
\begin{theorem}
\label{th:marton}
    Let $S \subset \F^{n}$ be a subset with $|S+S|\leq K |S|$. Then, $S$ can be covered with $(2K)^8$ transitions of a subspace $V\subset \F^{n}$ with $\abs{V} \leq \abs{S}$
\end{theorem}
Combining the above theorems yields the following corollary, which is the main tool we need from additive combinatorics.
\begin{corollary}
\label{cor:additive}
Let $S\subset \F^{n}$ with $\P[z_1, z_2\in S]{z_1 + z_2 \in S} \geq \epsilon$. There exists an \emph{affine} subspace $V\subset \F^{n}$ with $\abs{S \cap V} \geq \frac{\epsilon^{K_2}}{K_1} \abs{S}$ and $\abs{V}\leq \abs{S}$ where $K_1, K_2 \geq 0$ are absolute constants. One can choose $K_2 = 73$ and $K_1 = 3 \times 6^{72}$.
\end{corollary}
\begin{proof}
    Applying Theorem~\ref{th:bgs} to $S$, we obtain $S'$ with 1) $S'\subseteq S$ 2) $\abs{S'} \geq \frac{\epsilon}{3} \abs{S}$ 3) $\abs{S' + S'} \leq \pr*{\frac{6}{\epsilon}}^8 \abs{S} \leq \frac{3}{\epsilon}\pr*{\frac{6}{\epsilon}}^8 \abs{S'}$.
    Now applying Theorem~\ref{th:marton} to $S'$, we obtain a subspace $V$ with 1) $\abs{V} \leq \abs{S'} \leq \abs{S}$ and 2) $S'$ can be covered with $(2\frac{3}{\epsilon}\pr{\frac{6}{\epsilon}}^8)^8 = 6^{72} \times \epsilon^{-72}$ transitions of $V$. Thus, there exists a transition $z + V$ such that $\abs{(z+V)\cap S} \geq \abs{(z+ V) \cap S'} \geq \frac{\epsilon^{72}}{6^{72} }\abs{S'} \geq  \frac{\epsilon^{73}}{ 3  \times 6^{72} }\abs{S}$ as claimed.
\end{proof}

\section{Main result 1: Relating Gowers 3-norm to fidelity}
\label{sec:proof1}
This section proves our first main result, Theorem~\ref{th:main1} (re-stated below).
\begin{theorem*} [Restatement of Theorem~\ref{th:main1}]
    Let $\ket{\phi}$ be a quantum state with $\norm{\ket{\phi}}_{U^3}^{8} \geq \gamma $. Then, there exists a stabilizer state $\ket{s}$ such that $\abs{\braket{\phi}{s}} \geq \frac{\gamma^{C_2}}{C_1}$ where $C_1, C_2 > 0$ are two absolute constants.
\end{theorem*}

    Let $\ket{\phi} = \frac{1}{\sqrt{N}} g(x) \ket{x}$ with $\norm{\ket{\phi}}_{U^3}^8 = \norm{g}_{U^3}^8\geq \gamma$. We first prove if the result holds for the case where $g$ is a \emph{real} function, i.e., \emph{Claim~1} below holds, then the theorem holds for general complex-valued functions. We then prove \emph{Claim~1}.
    \paragraph{Claim~1.} If $g(x)\in \R$ for all $x$, then there exists a stabilizer state $\ket{s}$ with $\braket{s}{x} \in \R$ for all $x\in \F^n$ and $\abs{\braket{s}{\phi}} \geq \frac{\gamma^{C_2'}}{C_1'}$ for absolute constants $C_1' > 0$ and  $C_2' \geq 1$.
    
    We emphasize that the assumption that $g$ is real is essential in the proof of \emph{Claim~1} as Eq.~\eqref{eq:real} is not true for general $g$. We now show that \emph{Claim 1} implies our main theorem.
    \begin{proof}[Proof of Theorem~\ref{th:main1}]
        Assume that \emph{Claim~1} is true for absolute constants $C_1'$ and $C_2'$. If $g$ is an arbitrary complex function, we can write it as $g= g_R + i g_I$ where $g_R$ and $g_I$ are real functions. Also, define $\ket{\phi_R} = \frac{1}{\sqrt{N}} \sum_x g_R(x) \ket{x}$  and $\ket{\phi_I} = \frac{1}{\sqrt{N}} \sum_x g_I(x) \ket{x}$. By the triangle inequality for Gowers norm \cite{Gowers_2001} ($\norm{g}_{U^3} \leq \norm{g_R}_{U^3} + \norm{g_I}_{U^3}$), we have that either $\norm{g_R}_{U^3} \geq \frac{1}{2}\norm{g}_{U^3}$ or $\norm{g_I}_{U^3} \geq \frac{1}{2}\norm{g}_{U^3}$. Without loss of generality, we can assume that $\norm{g_R}_{U^3}^8 \geq \frac{1}{2^8}\norm{g}_{U^3}^8 \geq \frac{\gamma}{2^8}$. Furthermore, let us denote $\nu \eqdef \E[x]{\abs{g_R(x)}^2}$ where $0<\nu \leq 1$. We now set $\widetilde{g} = g_R/\sqrt{\nu}$ and $\ket{\widetilde{\phi}} = \frac{1}{\sqrt{N}}\sum_{x\in \F^n} \widetilde{g}(x)\ket{x}$, where $\widetilde{g}$ is a real function and
        \begin{align}
            \E[x\in\F^n]{\widetilde{g}(x)^2} = 1 \quad \text{and} \quad \norm{\widetilde{g}}_{U^3}^8 \geq \frac{\gamma}{2^8\nu^4}
        \end{align}
        
        Applying \emph{Claim 1} to $\widetilde{g}$, there exists a stabilizer state $\ket{s}$ with real coefficient in computation basis such that $\frac{1}{\sqrt{\nu}}\abs{\braket{s}{{\phi}_R}}= \abs{\braket{s}{\tilde{\phi}}} \geq \frac{1}{C_1'}\pr{\frac{\gamma}{2^8\nu^4}}^{C_2'}$. 
        Because $C_2'\geq 1$ and $\nu \leq 1$, we have $\abs{\braket{s}{{\phi}_R}} \geq \frac{\gamma^{C_2'}}{C_1'2^{8C_2'}}$
        Furthermore, since $\ket{s}$ has real coefficient, both $\braket{s}{\phi_R}$ and $\braket{s}{\phi_I}$ are real. As a result,
    \begin{align}
        \abs{\braket{s}{\phi}}^2 = \abs{\braket{s}{\phi_R} + i \braket{s}{\phi_I}}^2 = \abs{\braket{s}{\phi_R}}^2 + \abs{\braket{s}{\phi_I}}^2 \geq \abs{\braket{s}{\phi_R}}^2.    \end{align}
        Thus, Theorem~\ref{th:main1} holds for $C_2 = C_2'$ and $C_1 = 2^{8C_2'}C_1'$.
    \end{proof}

    We now turn to the proof of the \emph{Claim 1}. The main idea is to apply a random \emph{real} Clifford to balance the characteristic function of the state $\ket{\phi}$ so that we are able to use the proof technique for inverse Gower's norm theorem for Boolean functions.
    \paragraph{Claim~2.} There exists a Clifford $C\in \Cliff^R_n$ (see Subsection~\ref{sec:stab-form} for definitions) such that if $C\ket{\phi} = \frac{1}{\sqrt{N}} \sum_x \widetilde{g}(x) \ket{x}$, then $\E[x\in\F^n]{\widetilde{g}(x)^4} \leq 3$.

    \begin{proof}
       We apply a random Clifford to $\ket{\phi}$ chosen uniformly from $\Cliff_n^R$.
    Since $\Cliff_n^R$ is an orthogonal 2-design (Lemma~\ref{lem:cliff-design}), we have
    \begin{align}
        \E[C \in \Cliff^R_n]{ \sum_x \abs{\bra{x}C\ket{\phi}}^4}  = \E[O \in \mathcal{O}]{ \sum_x \abs{\bra{x}O\ket{\phi}}^4} = N \cdot \E[O] {\abs{\bra{x}O\ket{\phi}}^4} = \frac{3}{N + 2}
    \end{align}
    Therefore, there exists a real Clifford $C$ with 
    \begin{align}
        \sum_{x\in \F^n} \abs{\bra{x}C\ket{\phi}}^4 \leq \frac{3}{N+2} < \frac{3}{N}.
    \end{align}
    Defining $\widetilde{g}(x) = \sqrt{N} \bra{x} C\ket{\phi}$, we have
    \begin{align}
        \E[x\in \F^n]{\widetilde{g}(x)^4} = N \sum_{x\in \F^n}  \abs{\bra{x}C\ket{\phi}}^4 \leq 3.
    \end{align}
    This completes the proof of \emph{Claim~2}. 
    \end{proof}
    
    In the next claim, we relate the quantity $\E[x]{\abs{g(x)}^4}$ to the characteristic function of the quantum state.
    
    \paragraph{Claim~3.} Without loss of generality, we can assume that  $\sum_{\alpha} f(y, \alpha) \leq 3$ for all $y$.
    \begin{proof}
        Let $C$ be the real Clifford from \emph{Claim~2}. We consider $\ket{\psi} = C\ket{\phi} = \frac{1}{\sqrt{N}}\sum_{x\in\F^n} \widetilde{g}(x) \ket{x}$. First, note that  $\ket{\psi}$ has real coefficient in computational basis, and we have $F(\ket{\psi}) = F(\ket{\phi})$ and $\norm{\ket{\psi}}_{U^3} = \norm{\ket{\phi}}_{U^3}$ because stabilizer fidelity and Gower's norm are invariant under a Clifford operation. Therefore, if we prove \emph{Claim~1} for $\ket{\psi}$ it implies that the claim is true for $\ket{\phi}$. It is enough to show that for all $y\in \F^n$, we have $\sum_{\alpha \in \F^n} f_\psi(y, \alpha) \leq 3$. The inequality holds for $y = 0$ since
        \begin{align}
            \sum_{\alpha \in \F^n} f_\psi(0, \alpha)
            &= \sum_{\alpha \in \F^n} \abs{\widehat{\Delta_{0}\widetilde{g}}(\alpha)}^2\\
            &\stackrel{(a)}{=}\E[x]{\abs{\Delta_{0}\widetilde{g}(x)}^2}\\
            &= \E[x]{\abs{\widetilde{g}(x)}^4} \leq 3,
        \end{align}
        where $(a)$ follows from Parseval equality. Now define the subspace $V=\set{(0, \alpha): \alpha \in \F^n} \subseteq \F^{2n}$. By Lemma~\ref{lem:f-sum-ineq}, we have for arbitrary $y\in\F^n$
        \begin{align}
            \sum_{\alpha \in \F^n} f_\psi(y, \alpha) = \sum_{z \in V} f_\psi(z + (y, 0)) \leq \sum_{z\in V} f_\psi(z) = \sum_{\alpha \in \F^n} f_\psi(0, \alpha)\leq 3,
        \end{align}
        as desired.
    \end{proof}
    
    The proof from here uses the same approach as in \cite[Proof of Theorem~4.5]{hatami2019higher} with some modifications. We give the steps here, prove the first step requiring more adjustment, and defer other steps to the appendix.
    \paragraph{Claim~4.} Assume for a quantum state $\ket{\phi}$, $\norm{\ket{\phi}}_{U^3}^8 \geq \gamma$ and let $f$ be the corresponding characteristic function. Take $\delta = \frac{\gamma^2}{6}$ and $\epsilon =\frac{\gamma^2}{54} - O(2^{-n})$.  Then, there exists a function $\zeta:\F^n \to \F^n$ such that
    \begin{align}
        \P[y_1, y_2]{f(y_1, \zeta(y_1)) \geq \delta , f(y_2, \zeta(y_2)) \geq \delta, f(y_1 + y_2, \zeta(y_1 + y_2)) \geq \delta, \zeta(y_1) + \zeta(y_2) = \zeta(y_1 + y_2)} \geq \epsilon.
    \end{align}
    \begin{proof}
        For any $\zeta: \F^n \to \F^n$, we define 
        \begin{align}
            L(\zeta) \eqdef \P[y_1, y_2]{f(y_1, \zeta(y_1)) \geq \delta , f(y_2, \zeta(y_2)) \geq \delta, f(y_1 + y_2, \zeta(y_1 + y_2)) \geq \delta, \zeta(y_1) + \zeta(y_2) = \zeta(y_1 + y_2)}
        \end{align}
        We choose a random $\zeta$ and show that $\E[\zeta]{L(\zeta)} \geq \epsilon$. Let us denote $r(y) \eqdef \sum_{\alpha} f(y, \alpha)$. By \emph{Claim~3}, we can assume that $r(y) \leq 3$ for all $y$. We pick $\zeta(y)$ such that $\zeta(y) = \alpha$ with probability $\frac{f(y, \alpha)}{r(y)}$ independently for all $y$ (If $r(y) = 0$, pick $\zeta(y)=0$). Then, 
        \begin{align}
            \E[\zeta]{ L(\zeta)}
            &= \E[\zeta]{\P[y_1, y_2]{f(y_1, \zeta(y_1)) \geq \delta , f(y_2, \zeta(y_2)) \geq \delta, f(y_1 + y_2, \zeta(y_1 + y_2)) \geq \delta, \zeta(y_1) + \zeta(y_2) = \zeta(y_1 + y_2)}}\\
            &=\E[y_1, y_2]{\P[\zeta]{f(y_1, \zeta(y_1)) \geq \delta , f(y_2, \zeta(y_2)) \geq \delta, f(y_1 + y_2, \zeta(y_1 + y_2)) \geq \delta, \zeta(y_1) + \zeta(y_2) = \zeta(y_1 + y_2)}}.
        \end{align}
        Fix $y_1, y_2$ such that $y_1, y_2, y_1 + y_2$ are distinct with $r(y_1), r(y_2), r(y_1 + y_2) > 0$. By denoting
        \begin{align}
            \Lambda_{y_1, y_2}\eqdef\set{(\alpha_1, \alpha_2): f(y_1, \alpha_1) \geq \delta, f(y_2, \alpha_2) \geq \delta, f(y_1 + y_2, \alpha_1 + \alpha_2) \geq \delta}
        \end{align}
        we can re-write
        \begin{align}
            &\P[\zeta]{f(y_1, \zeta(y_1)) \geq \delta , f(y_2, \zeta(y_2)) \geq \delta, f(y_1 + y_2, \zeta(y_1 + y_2)) \geq \delta, \zeta(y_1 + y_2) = \zeta(y_1 + y_2)} \\
            &= \sum_{(\alpha_1, \alpha_2) \in \Lambda_{y_1, y_2}} \P[\zeta]{\zeta(y_1) = \alpha_1, \zeta(y_2) = \alpha_2, \zeta(y_1 + y_2) = \alpha_1 + \alpha_2}\\
            &=\sum_{(\alpha_1, \alpha_2) \in \Lambda_{y_1, y_2}} \frac{f(y_1, \alpha_1) f(y_2, \alpha_2) f(y_1 + y_2, \alpha_2 + \alpha_2)}{r(y_1) r(y_2) r(y_1 + y_2)}\\
            &\stackrel{(a)}{\geq} \frac{1}{27} \sum_{(\alpha_1, \alpha_2) \in \Lambda_{y_1, y_2}} {f(y_1, \alpha_1) f(y_2, \alpha_2) f(y_1 + y_2, \alpha_2 + \alpha_2)}
        \end{align}
        where $(a)$ follows since $r(y) \leq 3$ for all $y$. When one of $r(y_1), r(y_2), r(y_1+y_2)$ is zero, the above inequality still holds as both sides are zero. Because the probability that $y_1, y_2, y_1 + y_2$ are not distinct is $O(2^{-n})$, we have
        \begin{align}
            \label{eq:l-bound1}
            \E[\zeta]{ L(\zeta)} \geq \frac{1}{27} \E[y_1, y_2]{  \sum_{(\alpha_1, \alpha_2) \in \Lambda_{y_1, y_2}} {f(y_1, \alpha_1) f(y_2, \alpha_2) f(y_1 + y_2, \alpha_2 + \alpha_2)}} - O(2^{-n}),
        \end{align}
        Furthermore, we can bound the contribution of terms $\notin \Lambda_{y_1, y_2}$. To see this, observe
        \begin{align}
            \E[y_1, y_2]{\sum_{\alpha_1, \alpha_2: f(y_1, \alpha_1) < \delta} f(y_1, \alpha_1) f(y_2, \alpha_2) f(y_1 + y_2, \alpha_2 + \alpha_2)}
            &< \delta \E[y_1, y_2]{\sum_{\alpha_1, \alpha_2}  f(y_2, \alpha_2) f(y_1 + y_2, \alpha_1 + \alpha_2)}\\
            &< \delta 
            \cdot (\frac{1}{N} \sum_{z \in \F^{2n}} f(z))^2\\
            &= \delta.
        \end{align}
        Using the same argument for $y_2$ and $y_1 + y_2$, we obtain that
        \begin{align}
            \label{eq:l-bound2}
            \E[y_1, y_2]{\sum_{(\alpha_1, \alpha_2)\notin \Lambda_{y_1, y_2}} f(y_1, \alpha_1) f(y_2, \alpha_2) f(y_1 + y_2, \alpha_2 + \alpha_2)} < 3\delta. 
        \end{align}
        Combining Eq.~\eqref{eq:l-bound1} and Eq.~\eqref{eq:l-bound2}, we obtain
        \begin{align}
            \label{eq:l-bound3}
            \E[\zeta] {L(\zeta)} 
            &\geq \frac{1}{27} \E[y_1, y_2\in \F^n]{  \sum_{(\alpha_1, \alpha_2)} {f(y_1, \alpha_1) f(y_2, \alpha_2) f(y_1 + y_2, \alpha_2 + \alpha_2)}} -\frac{1}{9}\delta - O(2^{-n}).
        \end{align}
        We next relate the right-hand side of the above inequality to $\gamma$. By Eq.~\eqref{eq:f-additive}, we have
        \begin{align}
            \E[y_1, y_2\in \F^n]{  \sum_{(\alpha_1, \alpha_2)} {f(y_1, \alpha_1) f(y_2, \alpha_2) f(y_1 + y_2, \alpha_2 + \alpha_2)}} = \frac{1}{N}\sum_{z\in\F^{2n}} f^3(z).
            \label{eq:L-and-f3}
        \end{align}
        Moreover, we have
        \begin{align}
            \gamma &\stackrel{(a)}{\leq} \frac{1}{N}\sum_{z\in \F^{2n}}f^2(z) \\
            &=\frac{1}{N}\sum_{z\in \F^{2n}}f^{1/2}(z) f^{3/2} (z)\\
            &\stackrel{(b)}{\leq} \frac{1}{N}\sqrt{\sum_{z\in \F^{2n}}f(z)}\sqrt{\sum_{z\in \F^{2n}}f^3(z)}\\
            & \stackrel{(c)}{=} \sqrt{\frac{1}{N}\sum_{z\in \F^{2n}}f^3(z)}.
        \end{align}
        where $(a)$ follows from Eq.~\eqref{eq:gower-char}, $(b)$ follows from Cauchy-Schwartz inequality, and $(c)$ follows from Eq.~\eqref{eq:parseval-f}. Substituting this bound in \eqref{eq:l-bound3} and using \eqref{eq:L-and-f3}, we obtain
        \begin{align}
            \E[\zeta]{ L(\zeta)} \geq \frac{1}{27} \gamma^2 - \frac{1}{9}\delta - O(2^{-n}).
        \end{align}
        Using the definition of $\delta$ and $\epsilon$, we obtain $\E[\zeta]{ L(\zeta)} \geq \epsilon$.
    \end{proof}
    From here we set $\epsilon, \delta$ as in \emph{Claim~4} and $K_1, K_2$ as in Corollary~\ref{cor:additive}.
    \paragraph{Claim~5.}  Assuming the parameters in \emph{Claim~4}, there exists an ``affine'' linear map  $\ell: \F^n \to \F^n$ with 
    \begin{align}
        \sum_{z\in G(\ell)} f(z) \geq \delta \frac{\epsilon^{2K_2+2}}{K_1} N
    \end{align}
    where $G(\ell)\eqdef \set{(y, \ell(y)): y\in \F^n}$ is the graph of $\ell$.
    \begin{proof}
        Take $\zeta$ from \emph{Claim~4}. Let $S = \set{(y, \zeta(y)): f(y, \zeta(y)) \geq \delta}$. By \emph{Claim~4}, 
        \begin{align}
            \P[z_1, z_2\in S]{z_1 + z_2\in S} \geq \P[y_1, y_2 \in \F^n]{(y_1, \zeta(y_1)), (y_2, \zeta(y_2)), (y_1 + y_2, \zeta(y_1) + \zeta(y_2)) \in S} \geq \epsilon
        \end{align}
        By Corollary~\ref{cor:additive}, there exists an affine subspace $V\subset \F^{2n}$ with $\abs{S\cap V} \geq \frac{\epsilon^{K_2}}{K_1}\abs{S}$ for absolute constants $K_1, K_2 > 0$ and $\abs{V} \leq \abs{S} \leq N$. Let $U = \set{y: (y, \alpha) \in V}$ be the projection of $V$ onto first $n$ entries. Since $S$ has at most one element $(y, y')$ for each $y$ and $\abs{V \cap S} \geq \frac{\epsilon^{K_2}}{K_1} \abs{S}$, we have $\abs{U} \geq \frac{\epsilon^{K_2}}{K_1} \abs{S}$. By Lemma~\ref{lem:linear-map}, there exists affine linear map $\ell:\F^n\to \F^n$ with $\abs{G(\ell) \cap S} \geq \frac{\abs{ V\cap S} \abs{U}}{\abs{V}}$. Using bounds  $\card{V\cap S} \geq \frac{\epsilon^{K_2}}{K_1} \card{S}$, $\card{U} \geq \frac{\epsilon^{K_2}}{K_1} \card{S}$ and $\card{V} \leq N$. We obtain that
        \begin{align}
            \card{G(\ell) \cap S} \geq \frac{\epsilon^{2K_2} \card{S}^2}{K_1^2N}.
        \end{align}
        Since for all $z\in S$, we have $f(z) \geq \delta$ and $\card{S} \geq \epsilon N$,
        \begin{align}
            \sum_{z\in G(\ell)} f(z) \geq \sum _{z\in G(\ell) \cap S} f(z) \geq \delta \card{G(\ell) \cap S} \geq \delta\frac{\epsilon^{2K_2+2}}{K_1^2} N
        \end{align}
    \end{proof}
    \paragraph{Claim~6.} Assuming the parameters of \emph{Claim~4}, there exists a linear map  $\ell: \F^n \to \F^n$ with 
    \begin{align}
        \sum_{z\in G(\ell)} f(z) \geq \delta \frac{\epsilon^{2K_2+2}}{K_1^2} N.
    \end{align}
    \begin{proof}
        Take the \emph{affine} linear map $\ell$ from \emph{Claim~5}. By definition, we can write $\ell = \ell_0 + b$ where $\ell_0:\F^n \to \F^n$ is a linear map and $b\in \F^n$ is fixed. Note that $G(\ell) = G(\ell_0) + (0, b)$ and $G({\ell_0})$ is a subspace. Using Lemma~\ref{lem:f-sum-ineq}, we have
        \begin{align}
            \sum_{z\in G({\ell_0})} f(z) \geq \sum_{z\in G(\ell)} f(z), 
        \end{align}
        which implies that  \emph{Claim~5} holds for $\ell_0$ too.
    \end{proof}
   We now state three additional claims, whose proofs follow the same structure as those in \cite[Proof of Theorem~4.5]{hatami2019higher}. For completeness, we include these proofs in Appendix~\ref{sec:proof-claims}, adapted to our notation, and verify that the arguments hold even when our state is not a phase state.
    \paragraph{Claim~7.} Assuming the parameters of \emph{Claim~4}, there exists a \emph{symmetric} linear map $\ell:\F^n\to\F^n$ with 
    \begin{align}
        \sum_{z\in G(\ell)} f(z) \geq \delta^2 \frac{\epsilon^{4K_2+4}}{K_1^4} N
    \end{align}
    \paragraph{Claim~8.}Assuming the parameters of \emph{Claim~4}, there exists a \emph{symmetric} linear map $\ell:\F^n\to\F^n$ with  \emph{zero diagonal} and 
    \begin{align}
        \sum_{z\in G(\ell)} f(z) \geq \delta^2 \frac{\epsilon^{4K_2+4}}{K_1^4} N
    \end{align}
    \paragraph{Claim~9.} Assuming the parameters of \emph{Claim~4}, there exists a quadratic polynomial $q:\F^n \to \F$ such that 
    \begin{align}
        \abs{\E[x]{g(x) (-1)^{q(x)}}} \geq \delta \frac{\epsilon^{2K_2+2}}{K_1^2}
    \end{align}

We are now ready to present the proof of \emph{Claim~1.}

    \begin{proof}[Proof of Claim~1]
        We define the stabilizer state $\ket{s} = \frac{1}{\sqrt{N}} \sum_{x\in \F^n}(-1)^{q(x)}\ket{x}$ where $q$ is the quadratic map from \emph{Claim~9}. $\ket{s}$ has real coefficients in the computational basis. Additionally,
        \begin{align}
            \abs{\braket{s}{\phi}} = \frac{1}{N} \abs{ \sum_{x\in \F^n} (-1)^{q(x)} g(x) } \geq \delta \frac{\epsilon^{2K_2+2}}{K_1^2}.
        \end{align}
        Furthermore, we can choose $C_1' = 6 \times K_1^2 \times (54 + o(1))^{2K_2+2}$ and $C_2'=4K_2 + 6$.
    \end{proof}
\begin{remark}
\label{rem:constants}
    One can take in Theorem~\ref{th:main1}  $C_1 = 6  \times K_1^2 \times (54 + o(1))^{2K_2+2}\times 2^{32K_2 + 48}$ and $C_2 = 4K_2 + 6$ where $K_1, K_2$ are the constants from Corollary~\ref{cor:additive}. We did not attempt to optimize these constants in our proof. Furthermore, while  the constants $C_1$ and $C_2$ are very large ($\approx 6.85 \times 10^{1087}$  and $298$, respectively) for the present values of $K_1$ and $K_2$, we expect that future improvements on the bounds for $K_1$ and $K_2$ will reduce these constants.
\end{remark}

\section{Main result 2: Relating stabilizer rank to fidelity}
\label{sec:proof2}

We prove our second main result, Theorem~\ref{th:main2} (re-stated below), in this section.

\begin{theorem*}[Restatement of Theorem~\ref{th:main2}]
    For each $k>0$ there exists $\delta_k > 0$ such that the following hold. Let $\ket{\phi}$ be a quantum state with $\chi(\phi) \leq k$. Then, $F(\phi) \geq \delta_k$, where $\delta_k$ is a number that only depends on $k$.
\end{theorem*}
    For each positive integer $k$, we first introduce a set $\mathcal{G}_k$ of $k\times k$ matrices, which can be a Gram matrix of $k$ stabilizer states. More formally,
    \begin{align}
        \mathcal{G}_k \eqdef \set{G\in \C^{k\times k}: \exists n \text{ and } \ket{s_1}, \cdots \ket{s_k} \in \Stab_n: G_{ij} = \braket{s_i}{s_j}}
    \end{align}
    We also define
    \begin{align}
        \lambda^{*}_k \eqdef \inf_{G\in\mathcal{G}_k: \lambda_{\min}(G) \neq 0} \lambda_{\min}(G)
    \end{align}
    which is the infimum of the minimum eigenvalue of non-singular matrices in $\mathcal{G}$. The heart of our proof is the following lemma.
    \begin{lemma}
        \label{lm:lambda-min}
        For each $k$, we have $\lambda^{*}_k > 0$.
    \end{lemma}
    Before proving the above lemma, we use it to conclude our proof. 
    \begin{proof}[Proof of Theorem~\ref{th:main2}]
    
    Fix $k$ and a unit vector $\ket{\phi} \in \C^N$ for $N=2^n$ with $r = \chi(\phi) \leq k$. Thus, there exists $\ket{s_1}, \cdots, \ket{s_r}\in \Stab_n$ and $c_1, \cdots, c_r\in \C$ with $\ket{\phi} = \sum_{i=1}^{r}c_i \ket{s_i}$. To write this equality in matrix form, we introduce the following notation
    \begin{enumerate}
        \item $A\in \C^{N\times r}$ is $N\times r$ matrix whose $i^{th}$ column is representation of $\ket{s_i}$ in computational basis,
        \item $c\in \C^r$ is $(c_1, \cdots, c_r)$
        \item $\phi \in \C^N$ is the vector representation of $\ket{\phi}$ in computational basis
    \end{enumerate}
    Then, $\ket{\phi} = \sum_{i=1}^{r}c_i \ket{s_i}$ is equivalent to $Ac=\phi$. 
    Let $G$ be the Gram matrix of $\ket{s_1}, \cdots, \ket{s_r}$. We have $G=A^\dagger A$ and because $\ket{s_1}, \cdots, \ket{s_r}$ are linearly independent, $G$ is non-singular. Applying $G^{-1}A^\dagger$ to both sides of $Ac = \phi$, we obtain that 
    \begin{align}
    \label{eq:cginv}
        c = G^{-1}A^\dagger \phi
    \end{align}
    Furthermore, since $\ket{\phi}$ is unit vector, we have
    \begin{align}
        1 
        &= \phi^\dagger \phi \\
        &= \phi^\dagger A c \\
        &\stackrel{(a)}{=} \phi^\dagger A G^{-1} A^\dagger \phi\\
        &\leq \norm{G^{-1}} \norm{A^\dagger \phi}^2,\label{eq:unit}
    \end{align}
    where $(a)$ follow from Eq.~\eqref{eq:cginv}. $G\in \mathcal{G}_r$ and is non-singular. Hence, by Lemma~\ref{lm:lambda-min}, $\lambda_{\min}(G) \geq \lambda_r^*$ and $\norm{G^{-1}} \leq \frac{1}{\lambda_r^*}$. Combining it with Eq.~\eqref{eq:unit}, we obtain that $\norm{A^\dagger \phi} \geq \sqrt{\lambda_r^*}$. By Cauchy-Schwartz, we have $\norm{A^\dagger \phi}_\infty \geq \sqrt{\frac{{\lambda_r^*}}{k}}$. Note that the $i^{th}$ element of $A^\dagger \phi$ is $\braket{s_i}{\phi}$. Therefore, there exists $i$ such that $F(\phi) \geq \abs{\braket{s_i}{\phi}} \geq \sqrt{\frac{{\lambda_r^*}}{k}}$. We thus conclude that the claim of the theorem holds for 
    \begin{align}
        \delta = \min_{i\in[k]}\sqrt{\frac{{\lambda_i}}{k}}
    \end{align}
    \end{proof}

    We now prove Lemma~\ref{lm:lambda-min}
    \begin{proof}[Proof of Lemma~\ref{lm:lambda-min}]
        For the sake of contradiction, take a sequence $\set{G_n}_{n\geq 1}$ in $\mathcal{G}_k$ such that $G_n$ is non-singular for all $n$ and $\lim_{n\to\infty}\lambda_{\min}(G_n) = 0$. First, without loss of generality, we can assume that there exists $G^* \in \C^{k\times k}$ (not necessarily in $\mathcal{G}_k$) such that $\lim_{n\to\infty} G_n = G^*$. This is because for any $G\in\mathcal{G}_k$ and all $i,j\in[k]$, $\abs{G_{ij}} \leq 1$ and therefore, $\mathcal{G}_k$ is a subset of a compact set of all matrices. Consequently, any sequence in $\mathcal{G}_k$ has a convergent sub-sequence whose limit might be outside of $\mathcal{G}_k$. If $\set{G_n}_{n\geq 1}$ is not convergent, then we can replace it with one of its convergent sub-sequences, and it still satisfies two conditions of  $\set{G_n}_{n\geq 1}$. Thus, from now we assume that $\lim_{n\to\infty}(G_n)_{ij} = G^*_{ij}$ for all $i,j\in[k]$. If $\lambda_{\min}(G^*) > 0$, by Weyl inequality
         $\lambda_{\min}(G_n) \geq \lambda_{\min}(G^*) - \norm{G^* - G_n}$ which tends to $\lambda_{\min}(G^*) > 0$. This contradicts our initial assumption that $\lambda_{\min}(G_n)$ vanishes. Thus, it is enough to show that $\lambda_{\min}(G^*) > 0$.
        
        Note that $(G_n)_{ij}$ is an inner product of two stabilizer states. We state an analytical statement about the limit of such sequences in the following proposition.
        \begin{prop}
        \label{prop:analysis}
            Define
            \begin{align}
                \mathcal{I} = \set{x\in\C: \exists n \text{ and } \ket{s}, \ket{s'}\in \Stab_n x = \braket{s}{s'}}
            \end{align}
            Let $\set{c_n}_{n\geq 1}$ be a sequence in $\mathcal{I}$. If we have $\lim_{n\to\infty} c_n = c$, then, we have one of the following conditions:
            \begin{enumerate}
                \item $c\neq 0$ and for large enough $n$, $c_n = c$
                \item $c = 0$ and there exists a subsequence of $\set{c_n}$ such that all of its elements are zero
                \item $c=0$ and there exists a sub subsequence of $\set{c_n}$ such that all of its elements are non-zero
            \end{enumerate}
        \end{prop}
        We defer the proof of the above proposition to Appendix~\ref{apx:proof-prop}
        Since all entries of $G_n$ are in $\mathcal{I}$, Proposition~\ref{prop:analysis} implies that there exists a subsequence of $\set{G_n}_{n\geq 1}$ such as $\set{G_{n_k}}_{k\geq 1}$ such that for all $i,j\in[k]$ one of the following holds.
        \begin{enumerate}
            \item $(G_{n_k})_{ij} = G^*_{ij} \neq 0$ for all $k\geq 1$ 
            \item $(G_{n_k})_{ij} = G^*_{ij} = 0$ for all $k\geq 1$ 
            \item $\lim_{k\to \infty}(G_{n_k})_{ij} = G^*_{ij} = 0$ and $(G_{n_k})_{ij}\neq 0$ for all $k\geq 1$ 
        \end{enumerate}
        Without loss of generality we can set $\set{G_n}_{n\geq 1}$ to $\set{G_{n_k}}_{k\geq 1}$. So from now on, we assume that $\set{G_n}_{n\geq 1}$ satisfies the same condition. 

         We define a graph with vertex set $[k]$ and $i$ is connected to $j$ if they are in case 1. We claim that if there is a path between two vertices, then they \emph{cannot} be in case 3. To show this, consider a path $i_0, \cdots, i_\ell$ such that they are distinct vertices and $i_j$ is connected to $i_{j-1}$ for all $j\in[\ell]$. Consider the sequence of stabilizer state assigned to vertex $i_j$ (in the definition of $G_n$) and denote it by $\set{\ket{s_n^j}}_{n\geq 1}$. By applying a Clifford to all states, we can always assume that $\ket{s_n^0}$ is all zero state for all $n$. For all $j\in[\ell]$ and $n\geq 1$, we can write
         \begin{align}
             \ket{s_n^j} = \frac{1}{\sqrt{\abs{A_n^j}}}\sum_{x\in A_n^j} Q_n^j(x) \ket{x}
         \end{align}
         where $A_n^j$ is an affine subspace and $Q_n^j\in \mathcal{Q}$.
         We next use induction to show that $\abs{A_n^j} $ is bounded for all $j$. Note that
         \begin{align}
             \braket{s_n^j}{s_n^{j-1}} = \frac{1}{\sqrt{\abs{A_n^j}\abs{A_n^{j-1}}}} \sum_{x\in A_n^j \cap A_n^{j-1}} Q_n^j(x) Q_n^{j-1}(x)
         \end{align}
         Hence
         \begin{align}
             \abs{G^*{i_ji_{j-1}}} = \abs{\braket{s_n^j}{s_n^{j-1}}} \leq \frac{\abs{A_n^j \cap A_n^{j-1}}}{\sqrt{\abs{A_n^j}\abs{A_n^{j-1}}}} \leq
             \frac{\sqrt{\abs{A_n^{j-1}}}}{\sqrt{\abs{A_n^{j}}}}
         \end{align}
         which is equivalent to $\abs{A_n^{j}} \leq \frac{\abs{A_n^{j-1}}}{\abs{G^*{i_ji_{j-1}}}^2}$ as desired. We consider the connected components of the graph defined earlier $C_1, \cdots, C_r$, which is a partition of $[k]$. $G^*$ is a block matrix with respect to this partition because if $i,j\in[k]$ and $G^*_{i,j} \neq 0$, then there is an edge between $i$ and $j$ and therefore $i$ and $j$ belong to the same connected component. Furthermore, if we look at the block corresponding to $C_i$, it is the same as the corresponding block in $G_n$ for all $n$. This is because the only entries in $G^*$ that are not equal to $G_n$ are of case 3, and by our previous claim, they cannot exist within a connected component. Therefore, the $\lambda_{\min}(G^*)$ is the minimum of the minimum eigenvalue of all blocks, and they are non-zero because they are the principal submatrix of $G_n$ (for any $n$).  
\end{proof}

\appendix

\section{Proof details}
\subsection{Proof of \eqref{eq:parseval-f}}
\label{a:basic-calculations}

The proof is based on the basic application of Parseval's identity
\begin{align}
\begin{split}
    \frac{1}{N} \sum_{z\in\F^{2n}} f(z) &= \frac{1}{N} \sum_{y,\alpha}|\widehat{\Delta_y g}(\alpha)|^2\\
    &= \frac{1}{N^2} \sum_{y,\alpha}|\Delta_y {g(\alpha)}|^2 \hspace{2cm} \text{(Parseval)}\\
    &= (\frac{1}{N^2} \sum_{y,\alpha} |g(y)|^2 |g(y+\alpha)|^2)\\
    &= (\E[x]{\abs{g(x)}^2})^2\\
    &= \braket{\phi}{\phi}^2\\
\end{split}
\end{align}

\subsection{Proof of Lemma~\ref{lem:linear-map}}
\label{sec:proof-linear-map}
    First, we prove that there exists affine linear map $L:\F^n\to \F^n$ such that   $(y, L(y)) \in V$ for all $y\in U$. Let $U = u_0 + U_0$ where $U_0$ is a linear subspace of $\F^n$. Let $u_1, \cdots, u_k$ be a basis for $U_0$. For each $u_i$ choose $u_i'\in \F^n$ such that $(u_i +u_0, u_i')\in V$. We define $L(u_i + u_0) = u_i'$ which can be extended to an  affine map from $\F^n \to \F^n$ such that $(u, L(u))\in V$ for all $u\in \F^n$. 

    We claim that $\abs{V \cap (z+G(L))}$ is either zero or $\card{U}$. $\abs{V \cap (z+G(L))}$ cannot be larger than $\card{U}$. So it is enough to show that if $V \cap (z+G(L))$ is non-empty, then it has at least $\card{U}$ elements. Suppose that $z'\in {V \cap (z+G(L))}$ where $z' = z + (u', L(u'))$ for some $u' \in U$. Then, for any $u\in U$, 
    \begin{align}
        z + (u, L(u)) 
        &= (z + (u, L(u))) + (z' + z + (u', L(u')))\\
        &= (u, L(u)) + z' + (u', L(u')).
    \end{align}
    Since $(u, L(u)), z', (u', L(u')) \in V$ and $V$ is affine, then $z + (u, L(u)) \in V$. Therefore, for each $u\in U$, there is a unique element in $V \cap (z+G(L))$ as claimed. Because $G(L) + z$ are disjoint for distinct $z$, there are at most $\card{V}/\card{U}$ distinct $z$ such that $V \cap (z + G(L))$ is non-empty. Furthermore, 
    \begin{align}
        \card{S \cap V} = \sum_{z: V \cap (z + G(L)) \neq 0} \abs{S \cap V \cap (z + G(L))}
    \end{align}
    Thus, there exists $z \in \F^{2n}$ such that $\abs{S \cap V \cap (z + G(L))} \geq \card{S \cap V} \frac{\card{U}}{\card{V}}$. If $z = (z_1, z_2)$, then $z + G(L)$ is the graph of affine map $y \mapsto L(y+z_1) + z_2$, for which the desired result holds.

\subsection{Proof of Claims~7-9 in Section~\ref{sec:proof1}}
\label{sec:proof-claims}
We introduce some preliminary results in the beginning and then prove each claim. First, as in \cite{hatami2019higher}, we show that when $g$ is a real valued function $g : \F^{2n} \rightarrow \R$, then $f(y, \alpha) = 0$ if $\langle y, \alpha \rangle = 1$. To see this note that
\begin{align}
    \widehat{\Delta_y g}(\alpha) 
    &= \E[x]{g(x) \overline{g(x+y)} (-1)^{\langle x, \alpha \rangle}}\\
    &=\E[x]{g(x) {g(x+y)} (-1)^{\langle x, \alpha \rangle}}\label{eq:real}\\
    &= - \E[x]{g(x) {g(x+y)} (-1)^{\langle x+y, \alpha \rangle}}\\
    &= -\widehat{\Delta_y g}(\alpha).
\end{align}
(Note for arbitrary functions the above calculations show that the real part of corresponding entries will be zero). For a  linear map $\ell$, we define two functions from $\F^n \to \R$ as
\begin{align}
    R: y &\mapsto f(y, \ell(y)) \label{eq:r-def}\\
    H: y &\mapsto (-1)^{\langle y, \ell(y) \rangle} \label{eq:h-def}
\end{align}
First of all, since $f(y, \ell(y)) = 0$ when $H(y) = -1$, we have $R(y)H(y) = R(y)$ for all $y \in \F^n$. Next we prove a relation between the Fourier transform of $R$ and $\sum_{z\in G(L)} f(z)$ for appropriate affine function $L$. Note that that 
\begin{align}
    \widehat{R}(\alpha) 
    &= \E[y]{f(y, \ell(y))(-1)^{\langle y, \alpha\rangle}}\\
    &\stackrel{(a)}{=} \frac{1}{N} \sum_{y', \alpha'} f(y', \alpha')\E[y]{(-1)^{[(y', \alpha'), (y, \ell(y))]} (-1)^{\langle y, \alpha\rangle}}\\
    &=  \frac{1}{N} \sum_{y', \alpha'} f(y', \alpha')\E[y]{(-1)^{\langle y, \ell^t(y') + \alpha' + \alpha \rangle }}\\
    &= \frac{1}{N} \sum_{y', \alpha'} f(y', \alpha')\indic{\ell^t(y') + \alpha = \alpha'}\\
    &= \frac{1}{N} \sum_{y'} f(y', \ell^t(y') + \alpha)\\
    &= \E[y']{f(y', \ell^t(y') + \alpha)}.
\end{align}
For (a), we have used Lemma~\ref{lem:f-sum}. 
\begin{proof} [Proof of \emph{Claim~7}] We show that for any linear map $\ell:\F^n \to \F^n$, if 
\begin{align}
    \sum_{z\in G(\ell)} f(z) = \eta N,
\end{align}
then there exists a symmetric linear function $\ell':\F^n\to \F^n$ such that 
\begin{align}
    \sum_{z\in G(\ell')} f(z) \geq \eta^2 N.
\end{align}
Then, taking $\ell$ to be the linear function from \emph{Claim~6} completes the proof. Note that
\begin{align}
    \eta = \E[y\in \F^n]{R(y)} = \E[y]{H(y) R(y)} = \sum_\alpha \widehat{H}(\alpha)\widehat{R}(\alpha)  
\end{align}
Furthermore, for all $\alpha\in \F^n$, we have $\widehat{R}(\alpha) \geq 0$ and 
\begin{align}
    \sum_{\alpha\in \F^n} \widehat{R}(\alpha) = \sum_{\alpha} \E[y']{f(y', \ell^t(y') + \alpha)} = 1,
\end{align}
 Hence, $\widehat{R}(\alpha)$ forms a probability distribution over $\F^n$ and by Jenson's inequality
\begin{align}
    \sum_\alpha \widehat{H}^2(\alpha)\widehat{R}(\alpha) \geq \pr*{\sum_{\alpha\in\F^n} \widehat{H}(\alpha)\widehat{R}(\alpha) } ^2 = \eta ^2.
\end{align}
In addition, note that
\begin{align}
    \sum_\alpha \widehat{H}^2(\alpha)\widehat{R}(\alpha) = \E[y]{(H*H)(y) R(y)} = \E[y]{H(y) R(y) \indic{\ell(y) = \ell^t (y)}} = \E[y]{R(y)\indic{\ell(y) = \ell^t (y)}}
\end{align}
Let $\ell'$ be a symmetric linear map such that $\ell'(y) = \ell(y)$ on $\set{y: \ell(y) = \ell^t(y)}$ and arbitrary elsewhere. Then, we have
\begin{align}
    \sum_{y} f(y, \ell'(y)) \geq \sum_{y: \ell(y) = \ell^t(y)} f(y, \ell(y)) \geq N \eta^2,
\end{align}
as desired.
\end{proof}

\begin{proof} [Proof of \emph{Claim~8}]
    Take $\ell$ from \emph{Claim~7}. Let $v$ be the diagonal of $\ell$, i.e., $v_i = \langle e_i, \ell(e_i)\rangle$. Then, we can write $\ell' = \ell + v\langle v, \cdot \rangle$ where $\ell'$ has zero diagonal. Then, $\ell'(y) = \ell(y) + v \langle y, v \rangle $. If $\langle y, v \rangle = 0$, then $\ell(y) = \ell'(y)$. If  $\langle y, v \rangle = 1$, then $\langle y, \ell(y) \rangle = \langle y, v \rangle  = 1$, and thus, $f(y, \ell(y)) = 0$. Thus, for all $y$, $f(y, \ell(y)) \leq f(y, \ell'(y))$ and
    \begin{align}
        \sum_{y\in \F^n} f(y, \ell(y)) \leq \sum_{y\in \F^n}f(y, \ell'(y)).
    \end{align}
    which completes the proof.
\end{proof}
\begin{proof}[Proof of \emph{Claim~9}]
    Take $\ell$ from \emph{Claim~8} and define $H$ as in Eq.~\eqref{eq:h-def}. Note that by Eq.~\eqref{eq:fourier4}
    \begin{align}
        \sum_\alpha \widehat{Hg}(\alpha)^4 
        &= \E[y]{\pr{\E[x]{\Delta_y H(x) \Delta_y g (x)}}^2}\\
        &= \E[y]{\pr{\E[x]{(-1)^{\langle x, \ell (y)\rangle} \Delta_y g (x)}}^2}\\
        &= \frac{1}{N}\sum_y f(y, \ell(y))\\
        &\geq \delta^2 \frac{\epsilon^{4K_2+4}}{K_1^4}.
    \end{align}
    We also have 
    \begin{align}
        \sum_\alpha \widehat{Hg}(\alpha)^4 \leq \max_{\alpha}  \abs{\widehat{Hg}(\alpha)  }^2 \sum_{\alpha} \widehat{Hg}(\alpha)^2 = \max_{\alpha}  \abs{\widehat{Hg}(\alpha)  }^2.
    \end{align}
    As a result, there exists $\alpha \in \F^n$ such that
    \begin{align}
         \delta \frac{\epsilon^{2K_2+2}}{K_1^2} \leq \abs{\widehat{Hg}(\alpha)} = \abs{\E[x]{g(x) H(x) (-1)^{\langle \alpha,x\rangle}}}.
    \end{align}
    Thus, the claim is valid for $q(x) = \langle x, \ell x\rangle + \langle \alpha,x\rangle$.
\end{proof}

\subsection{Proof of Proposition~\ref{prop:analysis}}
        \label{apx:proof-prop}
        First assume that $c\neq 0$. If $\ket{s}, \ket{s'} \in \Stab_n$, then $\braket{s}{s'} = i^{\ell}2^{-m/2}$ for integer $m$ and $\ell = 0, 1, 2,3$. Therefore, $\mathcal{I} \setminus \mathcal{B}_r$ has finite element for each $r>0$. Take $r=\frac{\abs{c}}{2}$ and$\epsilon $ to be
        \begin{align}
            \min\pr{r, \min_{x\in \mathcal{I} \setminus (\mathcal{B}_{r}\cup \set{c})} \abs{ x - c}}
        \end{align}
        There exists $n_0$ such that for all $n\geq n_0$, we have $\abs{c_n - c} < \epsilon$. Since $\epsilon \leq r$, we have $\abs{c_n - c} < r = \frac{\abs{c}}{2}$. By triangle inequality, we get that $\abs{c_n} > \frac{\abs{c}}{2}$. Also, from the definition of $\epsilon$, $c_n$ cannot be any point in $\mathcal{I} \setminus (\mathcal{B}_{r}\cup \set{c})$. Therefore, $c_n = c$ for $n \geq n_0$.

        Now suppose that $c=0$. If there are finitely many $n$ such that $c_n =0$, then we can eliminate all of them to obtain a subsequence with non-zero elements. Otherwise, there is a subsequence of zero elements.

\section{Bell sampling}
\label{sec:bell-sampling}
Here we review the Bell sampling a procedure that can efficiently approximate a quantity closely related to Gowers 3-norm of a quantum state. For $y, \alpha \in \F^n$, let $W_{(y, \alpha)}\eqdef  i^{y\cdot \alpha}X^yZ^\alpha$ be the corresponding Weyl operator. Given two copies of a state $\ket{\phi}$, \emph{Bell sampling} is measuring the states in the Bell basis $\set{\ket{W_{z}} \eqdef W_z\otimes \one \ket{\Phi^+}: z \in \F^{2n}}$ where $\ket{\Phi^+}\eqdef \frac{1}{\sqrt{N}}\sum_{x\in \F^{n}}\ket{x}\ket{x}$ is the maximally entangled state on $2n$ qubits. Given four copies $\ket{\phi}^{\otimes 4}$, \emph{Bell difference sampling} is first Bell sampling on the first two copies and the last two copies to obtain $z_1, z_2 \in \F^{2n}$ and then outputting $z = z_1 + z_2$. By \cite[Eq~(3.1)]{gross2021schur}, if we measure two copies of $\ket{\phi}$ in basis of $W_z$, the output of measurements are the same with probability 
\begin{align}
    \frac{1}{2}\pr{ 1 + \sum_{z\in \F^{2n}}q(z) f(z)}
\end{align}
where $f$ is the characteristic function of $\ket{\phi}$ and $q = f * f$. Therefore, we obtain
\begin{lemma} \cite[Lemma~3.8]{arunachalam2024tolerant}
    We can estimate $\sum_{z\in \F^{2n}}q(z) f(z)$ up to an additive error $\delta$ using $O(1/\delta^2)$ copies and a circuit of size $O(n/\delta^2)$.
\end{lemma}

The quantity $\sum_{z\in \F^{2n}}q(z) f(z)$ above is related to the Gowers 3-norm by \cite[Eq.~(4)]{arunachalam2024tolerant}
\begin{align}
    \norm{\ket{\phi}}_{U^3}^{16} \leq \sum_{z\in \F^{2n}}q(z) f(z) \leq \norm{\ket{\phi}}_{U^3}^{8}
\end{align}

\section{Other relations between measures of stabilizer complexity}
\label{sec:other-rel}
In this section, we provide other relations between measures of stabilizer complexity. In particular, motivated by Theorem~\ref{th:main2}, we are interested on the relation of the following form.
    \begin{definition}
    \label{def:measures-rel}
        Let $\alpha, \beta$ be two measures defined on quantum states (which are arbitrary in general and not necessarily related to stabilizer states). Let us denote 
        \begin{align}
            \alpha^* = \sup_n\sup_{\ket{\phi}\in \C^N: \braket{\phi}{\phi} = 1} \alpha(\ket{\phi})\\
            \beta^* = \sup_n\sup_{\ket{\phi}\in \C^N: \braket{\phi}{\phi} = 1} \beta(\ket{\phi})
        \end{align}
        We say that $\alpha \longrightarrow \beta$ if for any $a < \alpha^*$, there exists a constant $b < \beta^*$ such that if $\alpha(\ket{\phi}) \leq a$, then $\beta(\ket{\phi}) \leq b$.
    \end{definition}
    With this notation, Theorem~\ref{th:main2} is equivalent to saying that $\chi \longrightarrow 1 - F$. In addition to three measures introduced in Subsection~\ref{sec:measures}, let us consider another measure called $T$-gate complexity that is useful for learning quantum state \cite{grewal2024efficientlearningquantumstates}. We set $\kappa(\ket{\phi})$ as the minimum number $k$ such that we can prepare state $\ket{\phi}$ with Clifford gates and  $k$ number of  $T$ gates.
    \begin{theorem}
    \label{th:rels}
        Given the relation defined in Definition~\ref{def:measures-rel}, we have
        \begin{itemize}
            \item $\chi \longrightarrow 1 - F, 1- \norm{\cdot}_{U^3}$
            \item $\kappa \longrightarrow \chi, 1 - F, 1- \norm{\cdot}_{U^3}$
            \item $1- F \longrightarrow  1- \norm{\cdot}_{U^3}$
            \item $  1- \norm{\cdot}_{U^3}  \longrightarrow 1 - F$
        \end{itemize}
        Furthermore, no other relations hold among these four measures (i.e., counter-examples exist for missing arrows). These relations are summarized in Table~\ref{tab:my_label}.
    \end{theorem}
    \begin{table}[]
        \centering
        \begin{tabular}{|c|c|c|c|c|}
        \hline
              & $~~\chi~~$ & $~~\kappa~~$ & $1-F$ & $1-\norm{\cdot}_{U^3}$ \\
                      \hline

            $\chi$ & & \cellcolor{red}& \cellcolor{green} &\cellcolor{green}\\
                    \hline

            $\kappa$ &\cellcolor{green} & & \cellcolor{green}&\cellcolor{green}\\
                    \hline

            $1- F$ &  \cellcolor{red}& \cellcolor{red}& & \cellcolor{green}\\
                    \hline

            $1-\norm{\cdot}_{U^3}$ &\cellcolor{red} &\cellcolor{red} & \cellcolor{green}&\\        \hline

        \end{tabular}
        \caption{Summarizing the results of Theorem~\ref{th:rels}: The color of the cell corresponding to row $\alpha$ and column $\beta$ is green if $\alpha \longrightarrow \beta$ and is red otherwise (i.e., counter-examples exist). }
        \label{tab:my_label}
    \end{table}
    \begin{proof}
        First note that $\longrightarrow$ is transitive, i.e., $\alpha \longrightarrow \beta$ and $\beta \longrightarrow \gamma$ implies $\alpha \longrightarrow \gamma$. $1- F \longrightarrow  1- \norm{\cdot}_{U^3}$ was proven in \cite{gross2021schur, arunachalam2024tolerant}. Our Theorem~\ref{th:main1} implies also $1- F \longrightarrow  1- \norm{\cdot}_{U^3}$. Therefore, using transitivity, stabilizer fidelity and Gowers 3-norm are equivalent and from here we only consider stabilizer fidelity. Our Theorem~\ref{th:main2} also implies that $\chi \longrightarrow 1 - F$. Furthermore, \cite[Lemma~2.5]{grewal2022low} implies that $\kappa \longrightarrow 1- F$. Also, because each $T$ gate can at most double the stabilizer rank we have $\kappa \longrightarrow \chi$. Moreover, for one qubit states, we have states that we cannot prepare with any number of $T$ gates (as the set of states that can be prepared is countable). On the other hand, for a one qubit state, the stabilizer rank is at most $2$ and stabilizer fidelity is at least $1/2$. This shows that $\chi \longrightarrow \kappa$ and $1-F\longrightarrow \kappa$ do not hold. Finally, take a sequence of states $\ket{\phi_n}$ whose stabilizer rank goes to infinity (this sequence exists by \cite{mehraban2023quadratic}). Now take $\ket{\psi_n} = \frac{1}{2}\ket{0} +c_n\ket{\phi_n}$ for appropriately chosen $c_n$. We have $\chi(\ket{\psi_n}) \geq \chi(\ket{\phi_n}) - 1$, so stabilizer rank of  $\ket{\psi_n}$ goes to infinity. Furthermore, their stabilizer fidelity is at least $1/4$ by looking at their overlap with $\ket{0}$. This shows that $1-F\longrightarrow \chi$ does not hold.
    \end{proof}

    \bibliographystyle{alphaurl}
\newcommand{\etalchar}[1]{$^{#1}$}
\end{document}